\documentclass[hidelinks]{article}
\usepackage[utf8]{inputenc}
\usepackage{csquotes}
\usepackage{hyperref}
\usepackage{listings}
\usepackage{amsmath} % Required for some maths elements         
\usepackage{amsthm}  % Required for some maths elements
\usepackage{amssymb} % Required for some maths elements
\usepackage{mathtools}
\usepackage[T1]{fontenc}
\usepackage{graphicx}
\usepackage{stmaryrd} %double brackets
\usepackage{multicol} %side by side equations
\usepackage{verbatim} %comments

\usepackage[section]{placeins}
\usepackage{enumitem} % alphabetical enumerate
\usepackage{natbib}
\usepackage{array}
\usepackage[toc,page]{appendix}

\usepackage{tikz} % graphs for probability spaces
\usetikzlibrary{shapes.geometric,calc,through, positioning}
\usepackage{tkz-euclide}
\tikzstyle{decnode} = [rectangle, minimum width=0.25cm, minimum height=0.25cm, draw=black]
\tikzstyle{natnode} = [circle, minimum width=0.25cm, minimum height=0.25cm, draw=black]
\tikzstyle{arrow} = [thick,->,>=stealth]

\newtheorem{theorem}{Theorem}
\newtheorem{lemma}{Lemma}
\newtheorem{proposition}{Proposition}

\theoremstyle{definition}
\newtheorem{definition}{Definition}[section]

\theoremstyle{remark}

\theoremstyle{definition}

\newcommand{\R}{\mathbb{R}}

\newcommand{\F}{\mathcal{F}}

\newcommand{\Pclass}{\mathcal{P}}

\newcommand{\Y}{\mathcal{Y}}

\newcommand{\W}{\mathcal{W}}

\newcommand{\Sent}{\mathcal{S}}

\title{A global accuracy characterisation of trust}
\author{Giacomo Molinari}

\begin{document}

\maketitle

\begin{abstract}
  \citet{dorst2021deference} put forward a deference principle called Total Trust, and characterise
  it in terms of accuracy: an agent totally trusts an expert iff they expect
  the expert to be more accurate than them. This note gives a new proof of their result
  using a global definition of accuracy due to \citeauthor{konek2023evaluating} (Forthcoming),
  rather than the local one used in the original. This allows for a simpler, direct proof of the global
  characterisation result.
\end{abstract}

\section{Notation and definitions}

Let $\W = \{w_1, ..., w_n\}$ a finite set of possible worlds. We can
capture the agent's opinions with a probability mass function $p : \Omega \to [0,1]$. A subset $A \subseteq \W$
called an \textit{event}, and I denote by $\F$ the set of all events. A probability mass function $p$ on
$\Omega$ induces a \textit{probability function} $Pr$ on $\F$, defined by:
\begin{equation}
  Pr(A) = \sum_{i=1}^n p(A)
\end{equation}
for every $A \in \F$.

A \textit{random variable} is a function $X : \W \to \R$, which we can think of as the following vector in $\R^n$:
\begin{equation}
  (x_1, ..., x_n) := (X(w_1), ..., X(w_n))
\end{equation}
We can think of a random variable $X$ as a \textit{gamble} on $\W$, where $X(w_i)$ is gain/loss 
resulting from the gamble if $w_i$ is the case.

Any probability mass function $p$ on $\W$ induces an expectation function over the set of all gambles:
\begin{equation}
  Exp_p(X) = \sum_{i=1}^n p(w_i)x_i
\end{equation}
Whenever $A \subseteq \W$ is an event, there is some gamble $I_A$ such that $I_A(w_i) = 1$ if $w_i \in A$,
and $I_A(w_i) = 0$ otherwise. Then if $Pr$ is the probability function induced by $p$, we have
$Pr(A) = Exp_p(I_A)$. Hence we can recover $Pr$ from $Exp_p$. I will abuse the notation and write $p$ for
both the probability mass function on $\W$, and the corresponding expectation function on $\R^n$, since
the latter just extends the former. Functions $f: \R^n \to \R$ are called \textit{previsions}, so $p$ will be
referred to as the agent's \textit{prevision}. The set of all previsions on $\R^n$ is denoted by $\Pclass$.
A prevision on $\R^n$ is \textit{coherent} iff it is the expectation function of some probability mass
function on $\W$.

If $A \subseteq \F$ and $p(A) \neq 0$, we can define a conditional probability mass function $p(\cdot |A)$ over $\W$,
which in turn induces conditional probability and expectation functions as above. Again I will denote this conditional
expectation function by $p(\cdot|A)$, and refer to it as the agent's \textit{conditional prevision}.

\section{A local accuracy characterisation of deference}

We are interested in deference principles, which specify what it means for an agent to defer to an expert.
I will use $\pi$ as the rigid designator of the agent's prevision, and $P$ as the definite description of the
prevision of a potential expert. That is, for every $w_i \in \W$, $P_i$ is the expert's prevision if $w_i$ is
the case. I will assume throughout that the agent's prevision $\pi$, and all the possible expert previsions $P_i$,
are coherent. I will write $[P(X) \geq 0]$ as shorthand for the event $\{w_i : P_i(X) \geq 0\} \subseteq \W$.
\citet{dorst2021deference} put forward the following local and global deference principles:
\begin{itemize}
\item \textbf{Total Trust (local)}\\
  For any random variable $X$, $\pi$ defers to $P$ with respect to $X$ iff:
  \begin{equation}
    p(X | \left[P(X) \geq t\right]) \geq t
  \end{equation}
  for any $t \in \R$ such that this conditional prevision is defined. If this is the case,
  we say that $\pi$ \textit{totally trusts} $P$ with respect to $X$.
\end{itemize}

\begin{itemize}
\item \textbf{Total Trust (global)}\\
  $\pi$ defers to $P$ on $X$ iff $\pi$ totally trusts $P$ with respect to every random variable $X$.
  That is, $\pi$ defers to $P$ iff:
  \begin{equation}
    \label{ttglobal}
    p(X | \left[P(X) \geq t\right]) \geq t
  \end{equation}
  for any $X: \W \to \R$ and $t \in R$ such that this conditional prevision is defined. Equivalently,
  $\pi$ defers to $P$ iff:
  \begin{equation}
    \label{ttglobal}
    p(X | \left[P(X) \geq 0 \right]) \geq 0
  \end{equation}
  for any random variable $X: \W \to \R$ such that this conditional prevision is defined.
\end{itemize}

One of the main results in \citep{dorst2021deference} is that $\pi$ totally trusts $P$ with respect to $X$
iff $\pi$ expects $P$ to be at least as accurate as itself with respect to $X$. To spell out the result, we
need to make this notion of accuracy precise by specifying a class of reasonable inaccuracy measures. For
any random variable $X$, a \textit{local inaccuracy measure} relative to $X$ is a function
$I_X : (\Pclass \times \W) \to \R$. When $p \in \Pclass$ and $w \in \W$, the value $I_X(p, w)$ quantifies
the inaccuracy of $p$ with respect to $X$ when $w$ is the case. Reasonable local inaccuracy measures are
characterised by the following local version of strict propriety: 
\begin{itemize}
\item \textbf{Strict propriety (local)}
  A local inaccuracy measure $I_X$ is strictly proper iff for any coherent prevision $p$, and any prevision $q$,
  we have:
  \begin{equation}
    p(I_X(p, \cdot)) \leq q(I_X(q, \cdot))
  \end{equation}
  with equality holding just in case $p(X) = q(X)$.
\end{itemize}

The accuracy characterisation of deference given in \citep{dorst2021deference} involves the local notion of Total
Trust, and the strictly proper local measures of inaccuracy defined above.
\begin{theorem}
  For any random variable $X$, $\pi$ totally trusts $P$ with respect to $X$ iff for every strictly proper local
  inaccuracy score $I_X$:
  \begin{equation}
    \pi(I_X(P, \cdot)) \leq \pi(I_X(\pi, \cdot))
  \end{equation}
\end{theorem}

To get a characterisation of (global) total trust, one needs to specify a reasonable class of global measures
of inaccuracy $I: (\Pclass \times \W) \to \R$, such that $I(p, w)$ summarises the inaccuracy of $p$ over all
random variables when $w$ is the case. Then one needs to show that global analogue of the theorem above holds for
such measures. While this shouldn't be hard to do given the results in \citep{dorst2021deference}, it's natural
to wonder whether there is any way to prove the global result directly. The next section will show this can be done,
using the characterisation of strictly proper global inaccuracy scores given by
\citeauthor{konek2023evaluating} (forthcoming).

\section{A global accuracy characterisation of deference}

\citeauthor{konek2023evaluating} (forthcoming) defines the inaccuracy of a prevision $p$ in terms of the gambles that an agent with this
prevision would find desirable. Recall that gambles are just a way to think about random variables, i.e.
vectors in $R^n$. Given a prevision $p : \R^n \to \R$, we can define its corresponding set of almost-desirable
gambles by:
\begin{equation}
  D_p = \{X \in \R^n : p(X) \geq 0\}
\end{equation}
These gambles are almost-desirable in the sense that their expected payoff according to $p$ is no lower than zero.

For any $w_i \in \W$, the \textit{ideal prevision at} $w_i$ is the one which assigns to each random variable $X$
its actual value $X(w_i) = x_i$. This clearly has to be the most accurate prevision at $w_i$. I abuse the notation and
write $w_i$ to denote both an element of $\W$ and its corresponding ideal prevision. The set of almost-desirable
gambles according to $w_i$ is:
\begin{equation}
  D_{w_i} = \{X \in \R^n : w_i(X) \geq 0\} = \{X : x_i \geq 0\} 
\end{equation}
This is just the set of all gambles whose payoff is non-negative at $w_i$. If $w_i$ is the actual world, then this
is the set of gambles that are actually almost-desirable.

There are two ways a prevision can be said to be inaccurate: it might find some gamble $X$ almost-desirable even
though $X$ is not actually almost-desirable (type 1 error); or it might find some gamble $X$ not almost-desirable
even though $X$ is actually almost-desirable (type 2 error). Each gamble on which $p$'s desirability
assessment differs from the actual one at $w_i$ will thus be a member of (only) one of the two \textit{error sets}:
\begin{flalign}
  \F^p_i  = \{X \in \R^n : X \in D_p \text{ and } X \notin D_{w_i}\} \\
  \Sent^p_i =  \{X \in \R^n : X \notin D_p \text{ and } X \in D_{w_i}\}
\end{flalign}
\citeauthor{konek2023evaluating} (forthcoming) defines his inaccuracy measures to take both kinds of error into account:
\begin{equation}
  I(p, w_i) = \int_{\F^p_i}\lvert x_i \rvert d\mu + \int_{\Sent^p_i}\lvert x_i \rvert d\mu
\end{equation}
and shows that $I$ is strictly proper when (i) $\mu$ is absolutely continuous w.r.t. the product Lebesgue measure,
(ii) $\mu$ gives positive measure to all open subsets of $R^n$, and (iii) $\mu(A) = \mu(-A)$ for every
$A \subseteq \R^n$, where $-A = \{-X \in \R^n : X \in A\}$.

Because the score is defined as an integral, it won't be sensitive to differences over sets of gambles with
Lebesgue measure zero. So it makes sense to introduce a variant of Total Trust that is ``up to sets of measure 0''.
\begin{definition}[Almost Everywhere Trust]
  $\pi$ trusts $P$ \textit{almost everywhere} iff it Totally Trusts $P$ on almost every $X$. That is,
  iff the following set:
  \begin{equation}
    \Y = \{X \in \R^n : \pi(\cdot| [P(X) \geq 0]) \text{ is defined, and } \pi(X | [P(X) \geq 0]) < 0\}
  \end{equation}
  has Lebesgue measure zero.
\end{definition}

The following lemma will be very useful for the other proofs in these notes. It shows that if Total Trust is
violated on some $X \in \R^n$, then it is violated some open subset $\Y$ of $\R^n$. This will
show that Total Trust is equivalent to Almost Everywhere Trust. Furthermore, $\Y$ can be picked
so that for every $Y \in \Y$ and every $w_i \in \W$, $P_i(Y) \neq 0$. This property ensures $\Y$
is symmetric in a sense which will be relevant for the proof of Proposition 2.
\begin{lemma}
  \label{lemma:1}
  Assume $\pi$ does not Totally Trust $P$. Then we can find an open subset $\Y$ of $\R^n$ such that:
  \begin{enumerate}
  \item For every $Y \in \Y$, $\pi(Y | [P(Y) \geq 0]) < 0$, and
  \item For every $Y \in \Y$ and every $w_i \in \W$, $P_i(Y) \neq 0$.
  \end{enumerate}
\end{lemma}
\begin{proof}
  We will construct an open subset of $\R^n$ that respects property (1), and then show that this set remains open
  once we remove from it all gambles that are assigned prevision 0 by some $P_i$. Assume $\pi$ does not Totally
  Trust $P$. So there is some gamble $X$ such that:
  \begin{equation}
    \label{TTviolation}
    \pi(X |[P(X) \geq 0]) < 0
  \end{equation}
  From this it follows that for any $\epsilon > 0$ sufficiently small, $\pi(X + \epsilon|[P(X) \geq 0]) < 0$. 
  Let $\lambda = \sup\{\epsilon > 0 : \pi(X + \lambda|[P(X) \geq 0]) < 0\}$. Since $\W$ is finite, we have:
  \begin{equation}
    [P(X) \geq 0] \equiv [P(X) \geq - \epsilon ] \equiv [P(X + \epsilon) \geq 0]
  \end{equation}
  for any $\epsilon > 0$ sufficiently small. Let
  $\xi = \sup \{\epsilon > 0 : [P(X) \geq 0] \equiv [P(X + \epsilon) \geq 0]\}$.
  %Then let $\lambda = \sup{\epsilon > 0 : } $
  %&\iff \pi(X | [P(X) \geq -\epsilon]) < 0 \text{ for any $\epsilon > 0$ small enough (since $\W$ finite)} \\
  %  &\iff \pi(X| [P(X + \epsilon) \geq 0]) < 0 \text{ for any $\epsilon > 0$ small enough} \\
  %  &\implies \pi(X + \epsilon| [P(X + \epsilon) \geq 0]) < 0  \text{ for any $\epsilon > 0$ small enough}
  %\end{flalign*} 
  Now let $\delta = \min\{\xi, \lambda\}$. Clearly $\delta > 0$. For every gamble $Z$ such that
  $0 < \inf(Z)$ and $\sup(Z) < \delta$, we have that $[P(X + Z) \geq 0]$ is
  equivalent to $[P(X) \geq 0]$, since $0 < P(Z)$ and $P(Z) \leq \sup(Z) < \delta \leq \xi$. So we have:
  \begin{flalign}
    \pi(X + Z | [P(X + Z) \geq 0]) =& \pi(X + Z|[P(X) \geq 0])\\
    &\leq \pi(X + \sup(Z) | [P(X) \geq 0])\\
    & < 0 \text{ (because $\sup(Z)< \delta \leq \lambda$)}
  \end{flalign}
  So $\pi$ does not trust $P$ on any gamble $(X + Z)$ where $0 < \inf(Z)$ and $\sup(Z) < \delta$.
  Define the the set $\Y_\delta$ as the set of all such gambles:
  \begin{equation}
    \Y_\delta = \{(X + Z) \in \R^n: 0 < \inf(Z) \text{ and } \sup(Z) < \delta\}
  \end{equation}
  This set is open w.r.t the product topology on $\R^n$, and Total Trust is violated on every element of $\Y_\delta$.
  Now for every $w_i \in \W$, let $\Theta_i = \{X : P_i(X) = 0\}$ the set of gambles with prevision 0
  according to expert $P_i$. Since the $P_i$'s are coherent previsions, the $\Theta_i$ are hyperplanes in $\R^n$,
  and are therefore closed w.r.t the product topology. So their complements $\Theta_i^c = \R^n {\sim} \Theta_i$
  are open. Now define:
  \begin{equation}
    \Y = \Y_\delta \cap \left(\bigcup_{i=1}^n \Theta_i^c \right) 
  \end{equation}
  Note that the union of the $\Theta_i^c$ is a union of open sets, and thus is open. Its intersection with the open set
  $\Y_\delta$ is therefore also open. The open set $\Y$ has property (1) because it's a subset of $\Y_\delta$,
  and has property (2) because it does not intersect any $\Theta_i$ by construction.
\end{proof}

The next proposition follows immediately from the above lemma. 
\begin{proposition}
  $\pi$ Totally Trusts $P$ iff $\pi$ Almost Everywhere Trusts $P$.
\end{proposition}
\begin{proof}
  If $\pi$ Totally Trusts $P$, then clearly $\pi$ Almost Everywhere Trusts $P$. For the other direction,
  assuming $\pi$ does not Totally Trust $P$, we can find an open subset $\Y$ of $\R^n$ where Total Trust is
  violated (Lemma \ref{lemma:1}). Since $\Y$ is open, it has positive Lebesgue measure. Thus $\pi$
  does not Almost Everywhere Trust $P$.
\end{proof}

Next we give a global characterisation of Total Trust in terms of accuracy.
\begin{proposition}
  $\pi$ Totally Trusts $P$ iff for any global strictly proper inaccuracy score $I$, $\pi$ expects $P$
  to be at least as accurate as itself, that is:
  \begin{equation}
    \label{expInacc}
    \pi(I(P, \cdot) - I(\pi, \cdot)) \leq 0
  \end{equation}
\end{proposition}
\begin{proof}
  The main work of the proof consists of expressing the left hand side of \eqref{expInacc} in terms of the
  integrals (over suitable domains) of the previsions $\pi(X[P(X) \geq 0])$ and $\pi(X[P(X) < 0])$.
  We can then use our assumptions about Total Trust to determine the sign of these integrals, and thus show
  that \eqref{expInacc} holds iff $\pi$ Totally Trusts $P$.

  We can rewrite $\pi(I(P, \cdot) - I(\pi, \cdot))$ as follows:
  \begin{flalign*}
    & \pi(I(P, \cdot) - I(\pi, \cdot)) = \\ %\sum_{i=1}^n \pi(w_i) \left( \int_{\E^{P_i}_i} \lvert x_i \rvert d\mu - \int_{\E^\pi_i} \lvert x_i \rvert d\mu \right) \\
    &= \sum_{i=1}^n \pi(w_i) \left( \int_{\F^{P_i}_i} \lvert x_i \rvert d\mu + \int_{\Sent^{P_i}_i} \lvert x_i \rvert d\mu - \int_{\F^\pi_i} \lvert x_i \rvert d\mu - \int_{\Sent^\pi_i} \lvert x_i \rvert d\mu\right) \\
    &= \sum_{i=1}^n \pi(w_i) \left( \int_{D_{P_i}{\sim} D_{w_i}} \lvert x_i \rvert d\mu + \int_{D_{w_i} {\sim} D_{P_i}} \lvert x_i \rvert d\mu - \int_{D_\pi {\sim} D_{w_i}} \lvert x_i \rvert d\mu - \int_{D_{w_i}{\sim}D_{\pi}} \lvert x_i \rvert d\mu\right) \\
  \end{flalign*}
  Note that when $X \in D_i$, we have $x_i \geq 0$, and when $X \notin D_i$, we have $x_i < 0$. So we can rewrite the
  last line above as:
  \begin{equation}
    \sum_{i=1}^n \pi(w_i) \left( - \int_{D_{P_i}{\sim} D_{w_i}} x_i d\mu + \int_{D_{w_i} {\sim} D_{P_i}} x_i d\mu + \int_{D_\pi {\sim} D_{w_i}} x_i d\mu - \int_{D_{w_i}{\sim}D_{\pi}} x_i d\mu\right)\\
    \label{presplit}
  \end{equation}
  Note that the domain of integration of the first integral in \eqref{presplit}, $D_{P_i}{\sim} D_{w_i}$,
  can be written as the union of two disjoint sets $D_{P_i} \cap D_\pi \cap D^c_{w_i}$ and $D_{P_i} \cap D^c_\pi \cap D^c_{w_i}$.
  The domain of the third integral can similarly be written as the union of $D_{\pi} \cap D_{P_i} \cap D^c_{w_i}$
  and $D_{\pi} \cap D^c_{P_i} \cap D^c_{w_i}$. So we have:
  \begin{flalign*}
    & - \int_{D_{P_i}{\sim} D_{w_i}} x_i d\mu + \int_{D_\pi {\sim} D_{w_i}} x_i d\mu \\
    &= - \int_{D_{P_i} \cap D_\pi \cap D^c_{w_i}} x_i d\mu  - \int_{D_{P_i} \cap D^c_\pi \cap D^c_{w_i}} x_i d\mu
    + \int_{D_{\pi} \cap D_{P_i} \cap D^c_{w_i}} x_i d\mu + \int_{D_{\pi} \cap D^c_{P_i} \cap D^c_{w_i}} x_i d\mu\\
    &= - \int_{D_{P_i} \cap D^c_\pi \cap D^c_{w_i}} x_i d\mu + \int_{D_{\pi} \cap D^c_{P_i} \cap D^c_{w_i}} x_i d\mu
  \end{flalign*}
  We can similarly rewrite the domains of the second and fourth integrals in \eqref{presplit} to show:
  \begin{equation*}
    \int_{D_{w_i} {\sim} D_{P_i}} x_i d\mu - \int_{D_{w_i}{\sim}D_{\pi}} x_i d\mu =  \int_{D_{w_i} \cap D_{\pi} \cap D^c_{P_i}} x_i d\mu - \int_{D_{w_i} \cap D_{P_i} \cap D^c_\pi} x_i d\mu
  \end{equation*}
  Together, these equalities allow us to rewrite the term within the sum in \eqref{presplit} as:
  \begin{equation*}
    - \int_{D_{P_i} \cap D^c_\pi \cap D^c_{w_i}} x_i d\mu + \int_{D_{\pi} \cap D^c_{P_i} \cap D^c_{w_i}} x_i d\mu 
    + \int_{D_{w_i} \cap D_{\pi} \cap D^c_{P_i}} x_i d\mu - \int_{D_{w_i} \cap D_{P_i} \cap D^c_\pi} x_i d\mu
  \end{equation*}
  By reordering the intersections we get:
  \begin{flalign*}
    &- \int_{D_{P_i} \cap D^c_\pi \cap D^c_{w_i}} x_i d\mu + \int_{D^c_{P_i} \cap D_\pi \cap D^c_{w_i}} x_i d\mu
    +  \int_{D^c_{P_i} \cap D_\pi \cap D_{w_i} } x_i d\mu - \int_{D_{P_i} \cap D^c_\pi \cap D_{w_i}} x_i d\mu\\
    &= - \int_{(D_{P_i} \cap D^c_\pi \cap D^c_{w_i}) \cup (D_{P_i} \cap D^c_\pi \cap D_{w_i})} x_i d\mu
    + \int_{(D^c_{P_i} \cap D_\pi \cap D^c_{w_i}) \cup (D^c_{P_i} \cap D_\pi \cap D_{w_i} )} x_i d\mu \\
    &= - \int_{D_{P_i} \cap D^c_\pi} x_i d\mu + \int_{D^c_{P_i} \cap D_\pi} x_i d\mu
  \end{flalign*}
  So we can rewrite \eqref{presplit} as:
  \begin{flalign*}
    &\sum_{i=1}^n \pi(w_i) \left( - \int_{D_{P_i} \cap D^c_\pi} x_i d\mu + \int_{D^c_{P_i} \cap D_\pi} x_i d\mu \right)\\
    &= \sum_{i=1}^n \pi(w_i) \left( - \int_{D^c_\pi} x_i \chi_{D_{P_i}}d\mu + \int_{D_\pi} x_i \chi_{D^c_{P_i}} d\mu \right)\\
    &= - \int_{D^c_\pi} \sum_{i=1}^n \pi(w_i) x_i \chi_{D_{P_i}}d\mu + \int_{D_\pi} \sum_{i=1}^n \pi(w_i)x_i \chi_{D^c_{P_i}} d\mu \\
  \end{flalign*}
  Note that if $X$ is a gamble for which $\pi([X \in D_P]) = 0$, then $\pi(w_i) = 0$ for all $i$ such that
  $w_i \in [X \in D_P]$. That is, $\pi(w_i) = 0$ for all $i$ such that $\chi_{D_{P_i}}(X) = 1$. So these
  gambles do not contribute to the first integral. In the right integral, we can similarly ignore the gambles
  for which $\pi([X \notin D_P]) = 0$. By removing them from the domain, we can rewrite the above as:
  \begin{flalign}
    & - \int_{X : \pi([X \in D_P]) \neq 0, X \notin D_\pi} \sum_{i=1}^n \pi(w_i) x_i \chi_{D_{P_i}}d\mu + \int_{X : \pi([X \notin D_P]) \neq 0, X \in D_\pi} \sum_{i=1}^n \pi(w_i)x_i \chi_{D^c_{P_i}} d\mu \\
    &= - \int_{X : \pi([X \in D_P]) \neq 0, X \notin D_\pi}  \pi(X[X \in D_P])d\mu + \int_{X : \pi([X \notin D_P]) \neq 0, X \in D_\pi} \pi(X[X \notin D_P])d\mu
  \end{flalign}
  So we have shown the following equality:
  \begin{flalign}
    \label{finalEq}
    &\pi(I(P, \cdot) - I(\pi, \cdot)) \\
    &= - \int_{X : \pi([X \in D_P]) \neq 0, X \notin D_\pi} \pi(X[X \in D_P])d\mu + \int_{X : \pi([X \notin D_P]) \neq 0, X \in D_\pi} \pi(X[X \notin D_P]) d\mu
  \end{flalign}

  For the left to right direction of the proposition, assume $\pi$ Totally Trusts $P$. Note that $[X \in D_P]$ is
  equivalent to $[P(X) \geq 0]$, and $[X \notin D_P]$ is equivalent to $[P(X) < 0]$. For every $X$ in the
  domain of integration of the first integral, the conditional prevision $\pi(\cdot | [X \in D_P])$ is defined,
  and $\pi(X[X \in D_P])$ has the same sign as $\pi(X | [X \in D_P]) = \pi(X | [P(X) \geq 0])$. And since $\pi$
  Totally Trusts $P$, this means $\pi(X | [X \in D_P]) \geq 0$ on every $X$ in the domain of integration, making
  the first integral non-negative. The same reasoning shows that the second integral is non-positive. Hence the whole
  expression on the right hand side of \eqref{finalEq} is no greater than zero. 

  For the other direction of the proposition, assume that $\pi$ does not Totally Trust $P$. Then
  there is some $X$ on which Total Trust is violated, i.e. $\pi(X | [P(X) \geq 0]) < 0$. If
  $X \notin D_\pi$, $X$ will be within the domain of integration of the first integral in \eqref{finalEq}.
  Let $\Y$ be the open set constructed as in Lemma \ref{lemma:1}, where we pick $\delta$ small enough to ensure 
  $\Y \subseteq D^c_\pi$. By construction, $\Y$ is also a subset of $\{X : \pi([X \in D_P]) \neq 0\}$, and
  $\pi(Y | [P(Y) \geq 0]) < 0$ for every $Y \in \Y$. So $\Y$ is contained in the domain of the first
  integral, and the integrand is strictly negative over $\Y$. Consider now the set $-\Y = \{-Y : Y \in \Y\}$.
  We have $-\Y \subseteq D_\pi$, since $\Y \subseteq D^c_\pi$. And by property (2) of Lemma \ref{lemma:1},
  for every $-Y \in \Y$ we have:
  \begin{equation}
    [Y \in D_P] \equiv [P(Y) \geq 0] \equiv [P(Y) > 0] \equiv [P(-Y) < 0] \equiv [-Y \notin D_P]
  \end{equation}
  And thus $\pi([-Y \notin D_p]) \neq 0$. So $-\Y$ is in the domain of the second integral of \eqref{finalEq}.
  Furthermore, for every $-Y \in -\Y$:
  \begin{flalign}
    \pi(-Y | [- Y \notin D_P]) > 0 &\iff \pi(Y | [-Y \notin D_P]) < 0\\
    &\iff \pi(Y | [Y \in D_P]) < 0 \\
    &\iff \pi(Y | [P(Y) \geq 0]) < 0
  \end{flalign}
  where the last condition holds because $Y \in \Y$. So the integrand is strictly positive over $-\Y$. Both $\Y$
  and $-\Y$ are open in $\R^n$, thus we can find a measure $\mu$ over $R^n$ such that: (i) $\mu$ is absolutely
  continuous w.r.t the Lebesgue measure, (ii) $\mu$ assigns positive measure to every open subset of $R^n$,
  (iii) $\mu(A) = \mu(-A)$ for every $A \subseteq \R^n$, and such that $\mu$ concentrates sufficient weight on
  $\Y$ (and therefore also $-\Y$) to make the right hand side of \eqref{finalEq} strictly positive.
  This $\mu$ defines a strictly proper global inaccuracy measure for which $\pi(I(P, \cdot) - I(\pi, \cdot)) > 0$.
   
  If $X \in D_\pi$, $X$ is in the domain of integration of the second integral
  in \eqref{finalEq}. Note that $\pi(X |[P(X) < 0])$ must be defined and strictly positive, for otherwise
  $\pi(X)$ would be strictly negative, contradicting $X \in D_\pi$. We can use
  a similar construction as the one in Lemma \ref{lemma:1} to find an open set $\Y \subseteq D_\pi$ such that
  $\Y$ is a subset of $\{X : \pi([X \notin D_P]) \neq 0, X \in D_\pi\}$, and $\pi(Y | [P(Y) < 0]) > 0$
  for every $Y \in \Y$. So $\Y$ is contained in the domain of the second integral in \eqref{finalEq}, and
  the integrand is strictly positive over $\Y$. Analogously to the previous case, the set $-\Y$ is contained
  in the domain of the first integral of \eqref{finalEq}, and the integrand is strictly negative on $-\Y$.
  As above, we can find $\mu$ that defines a strictly proper global inaccuracy measure for which
  $\pi(I(P, \cdot) - I(\pi, \cdot)) > 0$.
\end{proof}

\bibliography{list19}

\end{document}